\documentclass[cmp]{svjourmod}

\usepackage{dcolumn}
\usepackage{bm}
\usepackage{verbatim}       

\usepackage[dvips]{graphicx}
\usepackage{amssymb}
\usepackage{bm}
\usepackage{amsmath}
\usepackage{amsfonts}
\newcommand{\p}{\partial}

\newcommand{\dd}{{\rm d}}
\begin{document}
%

\title{Light cones in Finsler spacetime}


\author{E. Minguzzi}
\institute{Dipartimento di Matematica e Informatica ``U. Dini'', Universit\`a degli
Studi di Firenze,  Via S. Marta 3,  I-50139 Firenze, Italy \\
\email{ettore.minguzzi@unifi.it} }
\authorrunning{E. Minguzzi}

\date{}
\maketitle

\begin{abstract}
\noindent  Some foundational results on the geometry of Lorentz-Minkowski spaces and
Finsler spacetimes are obtained. We prove that the local light cone structure of a  reversible Finsler spacetime with  more than two dimensions is topologically the same as that of  Lorentzian spacetimes: at each point we have just two strictly convex causal cones which intersect only at the origin. Moreover, we  prove a reverse Cauchy-Schwarz inequality for these spaces and a corresponding reverse triangle inequality. The Legendre map is proved to be a diffeomorphism in the general pseudo-Finsler case provided the dimension is larger than two.
\end{abstract}

\section{Introduction}

While the Finsler generalization of Riemannian geometry has proved natural and successful, and is by now classical \cite{matsumoto86,bao00,szilasi14},
the Finslerian generalizations of Einstein's gravity and  Lorentzian geometry have met some important difficulties so far.


In order to fix the ideas let us consider a Finsler Lagrangian $L(x,v)$ which is not quadratic in the velocities but which has, in some domain, a Lorentzian Hessian with respect to the velocities. A choice could be given by the Finsler Lagrangian of Randers type
\begin{equation} \label{mmj}
L(x,v)= \frac{1}{2} [a(v_0^{2}-v_1^{2}-v_2^{2}-v_3^{2})^{1/2}+ b v_1]^2, \qquad a,b > 0.
\end{equation}
Unfortunately, this Lagrangian has some undesirable features since the Finsler metric $g_{\alpha \beta}=\p^2 L/\p v^\alpha \p v^\beta $ is not  defined outside the cone $v_0^{2} > v_1^{2}+v_2^{2}+v_3^{2}$.






In the literature there is no  consensus on how to introduce a notion of
Finsler spacetime.  We can roughly identify two approaches depending on whether a Finsler Lagrangian such as (\ref{mmj}) would be considered mathematically acceptable (although possibly not a solution to the field equations).


The first approach, which would accept (\ref{mmj}) as well posed, dates back to Asanov
who considered in his monograph \cite{asanov85} several mathematical
and physical aspects of a Finslerian generalization of Einstein's gravity.

Asanov worked with a
positive Finsler Lagrangian  defined only over a conical
subset $I^+ \subset TM$  interpreted as the subbundle of
future directed timelike vectors. Unfortunately, in Asanov's
approach there is no clear room for lightlike curves and lightlike
geodesics. Given the considerable development of causality theory
for Lorentzian geometry \cite{hawking73} and its many insights for
the global causal structure of spacetime this is certainly a
serious drawback of Asanov's definition.

In Finsler geometry  the metric is defined on the tangent bundle minus the zero section, that is, on the {\em slit tangent bundle} $TM\backslash0$. In the indefinite case, Asanov's idea of removing a portion of the tangent bundle,
so as to enlarge the family of   allowed Finsler Lagrangians
has been taken up by several authors, and is probably the most common approach which can be found in physical works. With reference to recent literature, Barletta and Dragomir \cite{barletta12} select a general open conic subbundle, Javaloyes and S\'anchez \cite{javaloyes13} do the same but impose some conditions in order to retain the null directions, L\"ammerzahl, Perlick and Hasse \cite{lammerzahl12} consider metrics defined almost everywhere so as to comprise in their framework a kind  of direct sum of Finsler metrics (Remark \ref{mxo}), while Pfeifer and Wohlfarth \cite{pfeifer11}  use a  more complex definition which involves Finsler Lagrangians positive homogeneous of arbitrary degree and the restriction to suitable subsets of the tangent bundle. Kostelecky \cite{kostelecky11} also removes some portion of the tangent bundle, apart from the zero section, where the Finsler Lagrangian or the metric become singular.

In these Asanov's type approaches the cone domain of $L$ is usually added  as an ingredient of the very definition of Finsler spacetime.



A  more restrictive definition of Lorentz-Finsler manifold was given by
Beem, and is essentially that adopted by mathematicians working in pseudo-Finsler geometry  \cite{beem70,beem76b,akbarzadeh88}. Here the Finsler Lagrangian is defined all over $TM\backslash 0$. This is the definition which we adopt in this work but we stress that some of our results, including the Finslerian reverse Cauchy-Schwarz inequality and the reverse triangle inequality, will hold as well for the above approaches, since the proof works on a conic subbundle.

Beem's definition is included in the next definitions.\footnote{We recall that $f\in C^{k,1}$ iff $f$ is continuously differentiable up to order $k$, with Lipschitz $k$-th derivatives. }

\begin{definition} \label{nod}
 A {\em
pseudo-Finsler manifold} $(M,L)$ is a pair given by a paracompact connected
$C^{3,1}$ manifold and a $C^{3,1}$  function
$L\colon TM\backslash 0 \to \mathbb{R}$, $(x,v)\mapsto L(x,v)$,
 positively homogeneous of second degree with respect to  the fiber
\[
L(x,sv)=s^2L(x,v), \quad \forall s>0,
\]
such
that the symmetric tensor $ g:
TM\backslash 0 \to T^*M\otimes_{M} T^*M$,
defined by\footnote{Here  $(x^\mu,v^\mu)$  are local canonical coordinates on $TM$ induced by a local coordinate system on $M$. Clearly, the definition is well posed, that is, independent of the coordinate system.}
\[
g=g_{ \mu \nu}\, \dd x^\mu \otimes \dd x^\mu=\frac{\p^2 L}{\p
v^\mu\p v^\nu} \,\dd x^\mu \otimes \dd x^\mu,
\]
is non-degenerate. The pseudo-Finsler manifold is called {\em reversible} if
$L{(x,v)}=L{(x,-v)}$. It is simply a  {\em Finsler manifold} if $g$ is positive definite and, finally, it is a {\em Lorentz-Finsler manifold} if $g$ has Lorentzian signature, i.e.\ $(-,$ $+, \cdots,+)$.

A pseudo-Finsler space for which $g$ does not depend on $x$ is a {\em pseudo-Minkowski space}  (a {\em Minkowski space} if $g$ is positive definite). A {\em Lorentz-Minkowski space} is a pseudo-Minkowski space for which $g$ has Lorentzian signature  $(-,+,\cdots,+)$. A pseudo-Finsler space for which $g$ does not depend on $v$ is a pseudo-Riemannian space. If $g$ does not depend on neither $x$ nor $v$ then it is a {\em pseudo-Euclidean space} (a {\em Euclidean space} in the positive definite case, and a {\em Minkowski spacetime} in the Lorentzian signature case).
%
\end{definition}
%
%
%
%

Sometimes we shall write $g_{(x,v)}$ in order to stress the dependence  on the base point $x$ and the fiber point $v\in T_xM\backslash 0$. We shall also write $g_v$ either in order to
shorten the notation or because we regard $v$ as an element of
$TM\backslash 0$.

\begin{remark}
From  the positive homogeneity condition on  $L$ it follows that $g$ is positive homogeneous of degree zero and satisfies
\begin{equation} \label{jui}
(a) \quad \frac{\p g_{(x,v)\, \mu \nu }}{\p v^\alpha} \, v^\nu=0,
\qquad (b) \quad \frac{\p g_{(x,v)\, \mu \nu }}{\p v^\alpha} \,
v^\alpha=0.
\end{equation}
Furthermore,
\begin{align}
\frac{\p L}{\p v^\mu}(x,v)&=g_{(x,v)\, \mu \nu } \,v^\nu,
\label{njr}\\
L(x,v) &=\frac{1}{2}\, g_{(x,v)\, \mu \nu} \,v^\mu v^\nu. \label{kip}
\end{align}
One could have equivalently defined a pseudo-Finsler space  as a pair $(M,g)$ where on $g$ are imposed the properties expressed by Eq.\ (\ref{jui}). In this approach the Finsler Lagrangian $L$ would be defined  through Equation (\ref{kip}). For this reason, with some abuse of notation, but  in order to be consistent with the notation for pseudo-Riemannian manifolds, we shall equivalently denote the pseudo-Finsler spaces  with $(M,g)$.
\end{remark}

%


\begin{remark} \label{ret}
The differentiability conditions on $L$ imply that $g$ is $C^{1,1}$.
Observe that $g$ is positive homogeneous of zero degree on the velocities and bounded on compact sets of the unit tangent bundle (built using an arbitrary Riemannian metric).  As a consequence, although $g$ cannot be extended to the zero section, $L$ and $\p L/\p v^\mu$ can be extended by setting: $L=0$,  $\p L/\p v^\mu=0$. With this definition $L\colon TM\to \mathbb{R}$ is $C^{1,1}$. A well known simple argument which uses positive homogeneity shows that $L$ is $C^2$ on the zero section if and only if $L$ is quadratic in the velocities (pseudo-Riemannian case).

By the Picard-Lindel\"of theorem the $C^{1,1}$ differentiability condition on $g$ is the minimal one which  guarantees the existence and uniqueness of geodesics \cite{minguzzi13d}.
The geodesic spray,  being constructed with terms of the form
$\p g_{v\, \alpha \beta}/\p x^\gamma v^\delta v^\eta$, is Lipschitz in $x$ and $C^{1,1}$ in the velocities (also on the zero section where it can be defined to vanish).
\end{remark}


According to Beem a Finsler spacetime  is a time oriented Lorentz-Finsler manifold, where the notion of time orientation will be introduced below. The Lagrangian  (\ref{mmj}) does not satisfy Beem's definition. Thus a first question is whether there are Finsler Lagrangians which satisfy Beem's definition but which are not quadratic. An affirmative answer is provided by the next example.

\begin{example} \label{exe}
Let us consider the reversible Finsler Lagrangian ($\alpha \ge 0$)
\begin{align}
L(x,v)&= \frac{1}{2} \left\{1-\alpha e^{-\frac{v_0^2}{v_1^2+v_2^2}-\frac{v_1^2+v_2^2}{v_0^2}} \right\} (-v_0^2+v_1^2+v_2^2), \label{dkb}
\end{align}
which is well defined and $C^4$ also at $v_0=0$ or on the $v_0$-axis, provided $L:=\frac{1}{2}(-v_0^2+v_1^2+v_2^2)$ there.
For $\alpha=0$ this is the Finsler Lagrangian of a Minkowski spacetime of dimension $2+1$.  By positive homogeneity the metric $g$ is determined by its value over the unit sphere $v_0^2+v_1^2+v_2^2=1$, and over this compact set it is certainly Lorentzian for sufficiently small $\alpha$. Thus, by continuity, we can conclude that  for sufficiently small $\alpha$ this Finsler Lagrangian determines a Lorentz-Finsler space in Beem's sense (actually a Lorentz-Minkowski space). A simpler 1+1 example is provided by  $L(x,v)=\frac{1}{2} (1-\alpha e^{-(\frac{v_0}{v_1})^2-(\frac{v_1}{v_0})^2} ) (-v_0^2+v_1^2)$. The continuity of the Hessian follows from the smoothness of the function $e^{-\frac{1}{x^2}-x^2}$ at the origin.
\end{example}

Compared with Asanov's, Beem's definition of Lorentz-Finsler space  is more
economic, for there is no need to introduce a convex subbundle among
the fundamental ingredients of a Finsler spacetime. This fact, at
least at the physical level, is quite important since we do not wish
to introduce too many dynamical fields among the gravitational
degrees of freedom.

A first aspect of Beem's definition which is often criticizes is
that of not being able to accommodate the typical electromagnetic
Lagrangian (i.e.\ there is no Beem's Lagrangian $L(x,v)$ such that
$F:=\vert L\vert^{1/2}$ is of Randers type) for $L$ would not be
differentiable at some point different from the origin. However,
this is certainly not a problem of Beem's definition of Finsler
spacetime, but rather a signal for those researchers who try to use
the Finsler formalism to unify gravity and electromagnetism that,
probably, such unification cannot be accomplished through such a
simple extension. This limitation in Beem's
definition is indeed desirable provided that we regard the notion of Finsler
spacetime as a step towards an extension of general relativity
into a natural and comprising gravitational theory. In this
sense the restrictions implied by  Beem's definition seem to
indicate some precise features of the Finsler generalization,
features that we wish to explore.

\begin{remark} \label{mxo}
In this connection the paper by L\"ammerzahl, Perlick and Hasse \cite{lammerzahl12}
 studies a low field Finslerian generalization of the Schwarzschild metric (see also \cite{li14}). To that end they introduce the notion of  static Finsler spacetime as the direct sum of a `space Finsler metric' with a temporal metric part. It is known, already in the positive definite case, that the direct product of Finsler spaces is not as natural as in Riemannian geometry \cite[p.\ 13]{chern05}. As a consequence, in their work these authors have to pass to their mentioned notion of Finsler spacetime in which $g$ is defined only almost everywhere.
  However, a Finsler Lagrangian such as (\ref{dkb}) should probably be considered as static since it is invariant under translations with respect to any direction, and is such that the distribution of planes $\textrm{Ker} ( g_{(x,a)\, \nu \mu} a^\nu \dd x^\mu)$ is integrable for every $a$. However, it is not static according to the definition of \cite{lammerzahl12} since the metric obtained from (\ref{dkb}) has non-vanishing off-diagonal components.
All this seems to suggest that space and time in a Finsler spacetime  are more tightly related than in Lorentzian geometry. This issue certainly deserves to be further investigated.
\end{remark}

Unfortunately, Beem's study does not clarify the causal structure of the theory, not even at a point. Let the subsets of $T_xM\backslash 0$
\begin{align*}
I_x&=\{v\colon g_{v}(v,v)<0\},\\
J_x&=\{v\colon g_{v}(v,v)\le 0\},\\
E_x&=\{v\colon g_{v}(v,v)=0\}=J_x\backslash I_x,
\end{align*}
 denote the sets of
{\em timelike}, {\em causal} and {\em lightlike} vectors at $x\in M$, respectively.
To start with we need to answer the next question which has been overlooked by previous researchers
\begin{itemize}
\item[] (i) Is it possible that some of the sets $I_x$,  $E_x$, $(T_xM\backslash 0)\backslash J_x$, be empty?
\end{itemize}
Some authors might have assumed a negative answer on the basis of the Lorentzianity of the manifold $(T_xM\backslash 0,g_v)$.
We shall see that it is possible to assign a Lorentzian metric
$g_v$ which satisfies all the assumptions but Eq.\ (\ref{jui}a) and
such that $I_x$ is empty (see Example \ref{pip}). Therefore the
intuitive expectation that the above sets are non-empty  requires
justification.

Beem's definition does not clarify the
structure of $I_x$. One of his results complemented with  an
observation by Perlick \cite{perlick06} shows that each
component of $I_x$ is convex, but several questions remained open:
\begin{itemize}
\item[] (ii) Can the components of $I_x$ have closures which intersect?  (iii) Can they be infinite in number? (iv) Does the number
of components depend on $x$? (v) Can they contain some line passing through the origin\footnote{Recall that we defined $I_x, J_x, E_x$ as subsets of $T_xM\backslash 0$, that is, we removed the origin, so the line in question (v) is really a line passing through the origin minus the origin itself.}
(infinite propagation speed in some direction)?
\end{itemize}
 and most importantly:
\begin{itemize}
\item[] (vi) Can there be
more or less than two components?
\end{itemize}
Beem tried to shed some light on these questions in the 2-dimensional case. He provided examples which show that the number of
connected components  must be and can be any multiple of 2 (4 in the
reversible case). This result was certainly interpreted negatively
 by theoretical physicists as a causal structure presenting say, six light
cones, cannot be easily interpreted.
Certainly, such a structure  does not approximate in any sense
the physical spacetime of our experience.

 It is clear that given
this situation the removal of the unwanted cones, and hence the
restriction  to a convex conic subbundle of $TM$ as in Asanov's
strategy could be considered a good compromise.

\begin{remark}
Actually, some authors   look for these multi-cone  features in
order to model  two or more  signal speeds (birefringence,
multi-refringence,  bimetric gravity theories)
\cite{skakala09,pfeifer11}. However, due to the result on the
convexity of the components of $I_x$ cited above, no two different
cones can be intersecting (for they would be the same component and
hence the same cone), so that, coming to their boundaries, they
cannot stay one inside the other.
\end{remark}

In order to solve the problem of the undesired cones, other authors,
including Ishikawa \cite{ishikawa81}, suggested that the causal
structure in Finsler geometry is actually Lorentzian. Given a
section $\sigma\colon M\to TM\backslash 0$, it is possible to construct a
Lorentzian metric $g_{\sigma(x)}$ and hence to define the causal
character of a vector $v\in T_xM\backslash 0$ through the sign of
$g_{\sigma(x)}(v,v)$. Of course in this approach one would have to
give a dynamics for $\sigma(x)$, and also, to be consistent, the
motion of free particles on $M$ would be given by geodesics on $(M,
g_{\sigma(x)})$ otherwise the causal character of the tangent vector
would not be preserved. However, he agrees with most authors that the motion of massive particles and light should be described by Finsler geodesics, namely by the stationary points of the functional $\int \!L(x,\dot x) \dd t$, a fact which makes his hybrid approach untenable.

It is clear that many of these interpretational problems would be
solved if we could show that in more than two spacetime dimensions
the number of components of $I_x$ is 2 (question (vi)), for this fact would show that  Beem's choices for the definitions of Finsler spacetime and for the notions of timelike, causal and lightlike  vectors are physically satisfactory.


For many purposes a  technical way to get rid of the undesired cones
consists in introducing and assuming  the existence of a global
timelike vector field $T$, so as to restrict oneself to those
timelike vectors $v$ which satisfy $g_{v}(v,T)<0$. Such vectors
would be called {\em future directed}. This is the approach followed
by Perlick \cite{perlick06} and Gallego Torrom\'e, Piccione and
Vit\'orio \cite{gallego12} in their study of Fermat's principle in
Finsler spacetimes. They face the following problem:
\begin{itemize}
\item[](vii) if $v$ and $T$ belong to the same component of $I_x$ can $v$ be
non future-directed, that is $g_{v}(v,T)\ge 0$?
\end{itemize}
This is important in order to provide a clear notion of observer.
They circumvent this difficulty observing that for the variational
purposes of their study the continuity of $g$ on $(x,v)$ guarantees that
$g_{v}(v,T)$ does not change sign in a sufficiently small
neighborhoods of the curve on $TM$. However, they observe that under  present
knowledge of the theory there is no reason to expect the sign to be
negative for any timelike $v$ in the same component of $T$.

Another difficulty is met in the introduction of an Hamiltonian
framework for geodesics. As we shall see  in Eq.\ (\ref{njr}), the
conjugate momenta is $p=g_{v}(v,\cdot)$. It is easy to show that due
to the non-degeneracy of $g$ the map $v\mapsto p$ is locally
injective. However, it is unclear whether it is  globally injective. While this map is
injective if $g_v$ does not depend on $v$ (Lorentzian manifold) this
is not obvious in the general case. Indeed, Perlick \cite{perlick06}
mentions this difficulty while Sk\'akala and Visser \cite{skakala09}
circumvent it  introducing the injectivity of the map $v\mapsto p$
in the very definition of Finsler spacetime. Thus we ask

\begin{itemize}
\item[](viii) Can there be $v_1,v_2\in T_xM\backslash 0$, $v_1\ne v_2$, such that
$g_{v_1}(v_1,\cdot)=g_{v_2}(v_2,\cdot)$?
\end{itemize}

In this work we will prove that all the above pathologies do not
really occur for reversible Finsler spacetimes with more that two dimensions, that is,
all the above questions have a negative answer. Actually, as we shall see, only question (vi) requires reversibility.

%
%

Given  two light cones we can decide, at least locally, to
call one {\em future} and the other {\em past}.  Furthermore, they
are strictly convex so that, as they do not contain lines passing through the origin (i.e.\ they
are {\em sharp}), they describe finite speed signals. The past and
future light cones are one the opposite of the other for reversible
spacetimes.

Questions (vii) and (viii)  will
be answered through a Finslerian generalization of the reverse
Cauchy-Schwarz inequality, where for (viii) it will be important the equality case. This inequality will also lead to a
Finslerian reverse triangle inequality.

Given these results we can more properly define
the Finsler spacetime as follows

\begin{definition} \label{gyh}
A Lorentz-Finsler manifold is {\em time orientable} if it admits
some continuous global timelike vector field.
\end{definition}

As in Lorentzian geometry, the global timelike vector field is used to provide a continuous selection of  timelike cones, the {\em future cone} being that containing the  vector field.

\begin{definition} \label{gyk}
A Finsler spacetime is a time-oriented Lorentz-Finsler manifold, namely one for which a continuous selection of {\em future} timelike cones has been made.
\end{definition}

According to this definition in a Finsler spacetime we can speak of {\em future cone} but not necessarily of {\em past cone}.
%
%
%
%
%
%
Of course, the main result of this work shows that, under reversibility and for  dimensions larger than two, it makes sense to speak of {\em past cone}.


The results of this work joined with those on local convexity proved
in  \cite{minguzzi13d} establish that a Finslerian causality theory
can be developed in much the same way as it has been done for the
usual one starting from Lorentzian geometry
\cite{hawking73}. In fact once some key results have
been established,  most proofs follow word by word from those for
the Lorentzian spacetime theory (already for $C^{1,1}$ metrics \cite{minguzzi13d}).

\begin{remark} \label{uyy}
 It must be remarked that if the Finsler spacetime is
not reversible then there are really {\em two} causality theories,
one for the  distribution of future cones and the other for the
distribution of past cones. In this respect Finsler spacetime
physics might presents bi-refringence features as it can accommodate the propagation of two
different fundamental signals (as in bimetric
gravity theory). The negative results obtained in \cite{perlick06,skakala09} seem to be related to the unnecessary restriction to  reversible metrics.
\end{remark}
%
%
%
%

\section{Finsler geometry at a point}
The results of this section will concern the geometry of
Lorentz-Finsler manifolds on a tangent space $V:=T_xM$, so we can
omit any reference to $x$. We shall often identify $T_vT_xM \simeq T_xM$  (see Remark \ref{ssj} for more details).

The results of this section hold for a general  Lorentz-Minkowski space $(V,L)$, where $L\colon V\to \mathbb{R}$ is  $C^{3,1}$. With $n+1$  we shall denote the dimension of  $V$.
%
%
As commented above we have the identities
\begin{align} \label{koi}
(i) \quad \frac{\p L}{\p v^\mu}(v)&=g_{v\, \mu \nu } v^\nu,
\qquad (ii)  \quad \frac{\p g_{v\, \mu \nu }}{\p v^\alpha} \, v^\nu=0,
\qquad (iii)  \quad \frac{\p g_{v\, \mu \nu }}{\p v^\alpha} \,
v^\alpha=0,
\end{align}
We have already defined the notions of timelike, causal, and
lightlike  vectors as subsets of $V\backslash \{0\}$ defined through the sign of $g_v(v,v)$. We have also defined the corresponding
subsets $I,J,E \subset V\backslash \{0\}$. For shortness we will write $V\backslash 0$ in place of $V\backslash\{0\}$.


The next remark gives more details on the relationship between pseudo-Finsler spaces and the pseudo-Minkowski geometry of their tangent spaces.

\begin{remark} \label{ssj}
Let $(M,L)$ be a pseudo-Finsler manifold and let $x\in M$, $v\in T_xM\backslash 0$.
Given a vector $b \in T_xM$, $b=b^\mu \p/\p x^\mu$, we can send it to a vector $b^{\vee}\in T_vT_xM$ through the {\em vertical lift} $\vee: T_xM \times T_xM\to TT_xM$, $(v,b) \mapsto (v, b^{\vee})$,  $b^{\vee}:=b^\mu \p/\p v^\mu $. The vertical lift  $\vee$ establishes a linear isomorphism between $T_xM$ and $T_vT_xM$.
The pullback function $(\vee^{-1})^*L\colon T T_xM \to \mathbb{R}$ allows us to introduce the pseudo-Minkowski space $(T_xM,(\vee^{-1})^*L)$. Analogously, the pullback metric $(\vee^{-1})^* g\colon T_xM\backslash 0 \to  T^*T_xM \otimes T^*T_xM$ is the quadratic form
\[
(\vee^{-1})^* g=g_{ \mu \nu}\, \dd v^\mu \otimes \dd v^\mu=\frac{\p^2 L}{\p
v^\mu\p v^\nu} \,\dd v^\mu \otimes \dd v^\mu,
\]
(the reader familiar with pseudo-Finsler geometry will recognize the  restriction of the Sasaki metric to $T_xM\backslash 0$.)  which allows us to introduce the pseudo-Riemannian manifold $(T_xM\backslash 0, (\vee^{-1})^* g_v)$. This is a formal trick, often made without mention, which simplifies the study of the metric $g$ at a point. For simplicity, and with some abuse of notation, we shall denote $(\vee^{-1})^*L$ and $(\vee^{-1})^* g$ with  $L$ and $g$, respectively.
\end{remark}

\subsection{Non-emptyness of spacelike and timelike sectors}
In this and the next sections we shall need the definition of
coordinate sphere. Let us choose a base on $V$ and let $\Vert
v\Vert=[(v^0)^2+(v^1)^2+\cdots +(v^{n})^2]^{1/2}$ be the Euclidean
norm induced by the base on $V$. Let $S:=\{v: \Vert v\Vert=1\}$ be
the coordinate unit sphere of $V$.

Let us consider question (i).

\begin{theorem}
The sets $I$,  $E$, $(V\backslash 0)\backslash J$, are non-empty.
\end{theorem}

\begin{proof}
Let $\hat{L}:=L\vert_S$ and let $\hat v_{min}$ and $\hat v_{max}$ be
points where $L$ attains its minimum and maximum value over the
compact set $S$. They must be stationary points of
$2L-\lambda(v\cdot v-1)$ where $\lambda$ is a Lagrange multiplier
and $``\cdot "$ is the scalar product induced by the coordinate
system. Thus using Eq.\ (\ref{koi}i) $g_{\hat
v_{m}}(\cdot,\hat{v}_{m})=\lambda \hat v_{m}\cdot$, where $\hat v_m$ is
any among $\hat v_{min}$ and $\hat v_{max}$ (observe that $\lambda$ cannot
vanish because $g_{\hat v_{m}}$ is non-singular). Let us identify
$V$ with $V^*$ through the scalar product $``\cdot "$, the
previous equation establishes that $g_{\hat v_{m}}$ regarded as an
endomorphism of $V$ has eigenvalue $\lambda$ and eigenvector $\hat v_m$.
Evaluating the previous equation on $\hat v_m$  we find $0\ne
\lambda=g_{\hat v_{m}}(\hat{v}_{m},\hat{v}_{m})=2L(\hat{v}_{m})$.
This means that we can find $\textrm{dim} V-1$ other eigenvectors of $g_{\hat{v}_{m}}$ which stay
in the hyperplane $P$ which is $\cdot\,$-orthogonal to $\hat{v}_{m}$.
We call $\check P$ the hyperplane passing through $\hat v_m$ and
parallel to $P$, hence tangent to $S$. Taking into account the Lorentzian signature of $g_v$, we have that if  $g_{\hat
v_{m}}(\hat{v}_{m},\hat{v}_{m})<0$ then all these eigenvectors have
positive eigenvalue, otherwise there is one with negative
eigenvalue.

Now suppose that $(V\backslash 0)\backslash J$ is empty, then $L\le 0$. Let
$\hat v_m=\hat v_{max}$, and joining the assumption with an  observation above
$L(\hat{v}_{max})<0$, thus $g_{\hat v_{max}}$ is positive definite on
$P$. Let $w\in P$, $w\ne 0$, then
\begin{align*}
L(\hat v_{max}+\epsilon w)=L(\hat v_{max})+\epsilon  g_{\hat
v_{m}}(\hat{v}_{m},w)+\frac{\epsilon^2}{2} g_{\hat v_{m}}(w,w)+o(\epsilon^2)
\end{align*}
Since $g_{\hat v_{m}}(\hat{v}_{m},w)=\lambda  \hat{v}_{m}\cdot w=0$,
for sufficiently small $\epsilon$, $v:=\hat v_{max}+\epsilon w\in
\check{P}$ is such that $0\ge L(v)>L(\hat v_{max})$.  Since $\Vert
v\Vert>1$, defined $\hat v=v/\Vert v \Vert$ we have
\[
0\ge L(\hat{v})=\frac{1}{\Vert v \Vert^2}\, L(v)\ge L(v)>L(\hat v_{max}),
\]
in contradiction with the
maximality of $\hat{L}$ at $\hat v_{max}$. Thus $(V\backslash 0)\backslash J\ne
\emptyset$.
%

Suppose that $I$ is empty, then $L\ge 0$. Let $\hat v_m=\hat v_{min}$, and
joining the assumption with an observation above
$L(\hat{v}_{min})>0$, thus $g_{\hat v_{min}}$ is  non-degenerate and
indefinite on $P$. Let $w\in P$, $w\ne 0$, be such that
$g_{\hat v_{min}}(w,w)<0$, then
\begin{align*}
L(\hat v_{min}+\epsilon w)=L(\hat v_{min})+\epsilon  g_{\hat
v_{m}}(\hat{v}_{m},w)+\frac{\epsilon^2}{2} g_{\hat v_{m}}(w,w)+o(\epsilon^2)
\end{align*}
Since $g_{\hat v_{m}}(\hat{v}_{m},w)=\lambda  \hat{v}_{m}\cdot w=0$,
for sufficiently small $\epsilon$, $v:=\hat v_{min}+\epsilon w\in
\check{P}$ is such that $0\le L(v)<L(\hat v_{min})$. Since $\Vert
v\Vert>1$, defined $\hat v=v/\Vert v\Vert$, we have $L(\hat{v})\le L(v)<L(\hat v_{min})$, in contradiction with the
minimality of $\hat{L}$ at $\hat v_{min}$. Thus $I\ne \emptyset$. By
continuity of $\hat{L}$ on $S$, $E$ is non-empty. $\square$
\end{proof}

\begin{example} \label{pip}
It is possible to construct a Lorentzian manifold $(V\backslash 0,
g_v)$ for which for every $s>0$, $g_{sv \, \mu \nu}=g_{v \mu \nu}$,
but $J=\emptyset$. Let $V=\mathbb{R}^4$ with coordinates
$(v^0,v^1,v^2,v^3)$, let $r=\Vert v\Vert$, and let $S^3=\{v:\Vert
v\Vert=1\}$ be the unit sphere. Let us introduce the Hopf fibration
$S^3\to S^2$ with the usual coordinates \cite{trautman84}
$(\psi,\theta,\phi)$, then
\[
g_v=\dd r^2+r^2[-(\dd \psi-\cos \theta \dd \phi)^2+\dd \Omega^2]
\]
satisfies the mentioned conditions. It does not contradict the above
theorem because Eq.\ (\ref{koi}i) is not satisfied (to see this the
metric must be rewritten in the coordinates $\{v^\alpha\}$), so it
cannot be obtained from a Finsler Lagrangian.
\end{example}

\subsection{Convexity of the causal sublevels of $L$}

Let us define for $c\ge 0$ the following subsets of $V\backslash 0$
\begin{align*}
J(c)&:=\{ v\in V\backslash 0\colon g_v(v,v)\le -c^2\},\\
I(c)&:=\{ v\in V\backslash 0\colon g_v(v,v)< -c^2\}, \\
E(c)&:=\{ v\in V\backslash 0\colon g_v(v,v)= -c^2\},
\end{align*}
so that $J(0)=J$, $I(0)=I$, $E(0)=E$.

Let us recall that a subset $S\subset V\backslash 0$ is {\em connected} if
regarded as a topological space with the induced topology it is not
the union of two disjoint open sets. Equivalently $S$ is connected
if there are no subsets $H,K\subset S$ such that $S=H\cup K$ and
$H\cap \bar{K}=\bar{H}\cap K=\emptyset$ (the closure is in the
topology of $S$) \cite[Theor. 26.5]{willard70}. The relation $x\sim
y$ if $x$ and $y$ belong to a connected subset of $S$, is an
equivalence relation. Its equivalence classes are connected sets
called  the {\em components} of $S$. Every maximal connected subset
is a component \cite{willard70}. The continuous image of connected
sets is connected.

For every $s>0$ the homothety $v\mapsto sv$ is a homeomorphism and
hence sends components of $J(c)$, $I(c)$, $E(c)$,  into components
of $J(sc)$, $I(sc)$, $E(sc)$,  respectively. As it is invertible
this correspondence between components of, say, $J(c)$ and $J(c')$
for $c,c'>0$, is actually a bijection. As a consequence, one could
restrict the attention to the  components of $J(1)$, $I(1)$, $E(1)$,
$J$, $I$, $E$. Different components will be distinguished through  an index, e.g.\
$I^\alpha$.

\begin{remark}
A set made by two closed cones intersecting just at the origin of $V$ is connected according to the topology of $V$. Since we use the topology of $V\backslash 0$ the  corresponding notion of 'component' is non-trivial.
%
%
%
%
\end{remark}

\begin{lemma} \label{juz}
For every $c>0$ the set $E(c)$ is a $C^{3,1}$ orientable imbedded
codimensional one submanifold (hypersurface) of $V\backslash 0$.  We have for every $c$, $\p
I(c)=\p J(c)=E(c)$, $\overline{I(c)}=J(c)$, $\textrm{Int}\,
J(c)=I(c)$, where we use the topology of $V\backslash 0$. In particular,
\[
\p I=\p J=E, \quad \overline{I}=J, \quad \textrm{Int}\, J=I.
\]
\end{lemma}

\begin{proof}
The sets $E(c)$ are the level sets of $L$ which is $C^{3,1}$. Its differential
 $g_v(v,\cdot)$ is non-vanishing for $v\ne 0$ thus the implicit function theorem
implies the first claim \cite[Cor.\ 8.10]{lee03}. Observe that since
$g_v(v,\cdot)\ne 0$, $L$ is larger than $-c^2$ on one side and
smaller on the other, thus $E(c)$ is also oriented. This observation
and the continuity of $L$ imply $\p I(c)=E(c)$,
$\overline{I(c)}=J(c)$ and $\textrm{Int} \,E(c)=\emptyset$. From here
the other equations follow easily in particular the last one is
obtained setting $c=0$. $\square$
\end{proof}
%
%
%
%
%
%
%

The following result improves a previous result by Beem and
Perlick\footnote{Beem seems to use a theorem by Hadamard which
characterizes the convexity of a set through the curvature of its
boundary. However, it is unclear to this author which version of the
theorem Beem is using since Hadamard's  theorem holds for sets
having a compact boundary.} \cite{beem70,perlick06}.

\begin{theorem} \label{jod}
For $c>0$ the connected components of $J(c)$ are strictly convex.

For $c=0$ the connected components $J^\alpha$ of $J$ are convex, moreover any
non-trivial convex combination of two vectors in $J^\alpha$ returns
a timelike vector in $\textrm{Int}\,  J^\alpha$  unless the vectors
are both proportional to a lightlike vector.\footnote{In Corollary
\ref{sha} we shall establish that they have the same orientation.}

For $c\ge 0$ the connected components of $I(c)$ are convex.
\end{theorem}

\begin{proof}
According to a theorem due to Titze and Nakajima \cite{burago74} a
closed connected set $F\subset \mathbb{R}^n$ is convex if and only
if it is locally convex, in the sense that for every $p\in F$ there
is some convex neighborhood $U\ni p$ such that $U\cap F$ is convex.
Thus let $c>0$ and let $v \in J^\gamma(c)$, where $J^\gamma(c)$ is a
component of $J(c)$. Since $g_v(v,v)\le -c^2<0$, $v\in I$, and since
$I$ is open there is a convex neighborhood $U\ni v$ such that
$U\subset I$ and $0\notin U$. Suppose by contradiction that
$J^\gamma(c)\cap U$ is not convex, then there are
 $v_1,v_2\in I$, $v_1\ne v_2$, which
belong to this set but such that $w(\alpha):=(1-\alpha)v_1+\alpha
v_2$ does not belong to it for some $\bar{\alpha} \in(0,1)$.
However, observe that the whole segment $w([0,1])$ belongs to $U$
since this set is convex. Then there is some $\beta\in (0,1)$ such
that $w(\beta)$ maximizes
$f(\alpha):=2L(w(\alpha))=g_{w(\alpha)}(w(\alpha),w(\alpha))$.
Clearly $f(\beta)\ge f(\bar{\alpha})> -c^2$. Then Taylor expanding
$f$ at $\beta$, with $\Delta \alpha=\alpha-\beta$ the linear term
vanishes because $\beta$ is a maximum
\begin{equation} \label{jlo}
2\frac{\p L}{\p v^\gamma} \frac{\p w^\gamma(\alpha)}{\p
\alpha}\vert_{\alpha=\beta}=2g_{w(\beta)}(w(\beta),v_2-v_1)=0,
\end{equation}
 and the Taylor expansion reads
\begin{equation} \label{joo}
f(\beta+\Delta\alpha)=f(\beta)+ \Delta \alpha^2
g_{w(\beta)}(v_2-v_1,v_2-v_1)+o(\Delta \alpha^2).
\end{equation}
 However, Eq.\
(\ref{jlo}) shows that $v_2-v_1$ is $g_{w(\beta)}$-spacelike (in the
language of Lorentzian geometry), as it is orthogonal to the
$g_{w(\beta)}$-timelike vector $w(\beta)\in U\subset I$. Finally,
Eq.\ (\ref{joo}) gives a contradiction since the second order term
is positive, and hence $\beta$ cannot be a maximum. This argument
proves the convexity of $J^\gamma(c)$; its strict convexity follows
from the fact that the contradiction is reached also under the
weaker assumption $f(\beta)\ge -c^2$ which proves that no point in
the interior of the connecting segment can stay in $E(c)$. This fact
proves also that $I^\gamma(c)$ is convex.

Let us consider the case $c=0$. Since the map on $I$, $v \to cv$,
for given $c\ne 0$, is a homeomorphism, the components of $J(1)$ are
in one to one correspondence through this map with the components of
$J(c)$. Thus given a component $J^\alpha(1)$, the set
$J^\alpha(c):=cJ^\alpha$ is a component of $J(c)$. In particular it
contains $J^\alpha(c')$ for $c<c'$ and no other component of
$J(c')$. Thus $I^\alpha=\cup_{c>0} J^\alpha(c)$ is a component of
$I$ which is convex as any two points belong to some $J^\alpha(c)$.
The closure of a convex set is convex thus
$J^\alpha:=\overline{I^\alpha}$ is convex. Observe that $J^\alpha$
is connected because the closure of a connected set is connected.
Let
 $v_1,v_2\in J^\gamma$, $v_1\ne v_2$, by convexity $w(\alpha):=(1-\alpha)v_1+\alpha
v_2$ belongs to $J^\gamma$. Suppose that for some $\beta \in (0,1)$,
$w(\beta)$ does not belong to $I^\gamma$, then $f(\beta)=0$ and
hence $\beta$ is a maximum of $f$. Arguing as above
$g_{w(\beta)}(w(\beta),v_2-v_1)=0$ which proves that $v_2-v_1$ is
$g_{w(\beta)}$-orthogonal to the lightlike vector $w(\beta)$ and
hence either $g_{w(\beta)}$-spacelike or $g_{w(\beta)}$-lightlike and proportional
to $w(\beta)$. The Taylor expansion (\ref{joo}) shows that the
former possibility cannot hold. Observe that
$v_1=w(\beta)-\beta(v_2-v_1)$, $v_2=w(\beta)+(1-\beta)(v_2-v_1)$,
thus they are both proportional to the lightlike vector $w(\beta)$. $\square$
\end{proof}

Let $\hat{J}=J\cap S$, $\hat{I}=I\cap S$ and $\hat{E}=E\cap S$. The
next result answers questions (ii) and (iii).

\begin{proposition}
Different components of $\hat{E}$ are disjoint (light cones do not
touch) and they are finite in number.  The components of $\hat{I}$ have topology
$D^{n}$ while the components of $\hat{E}$ have topology $S^{n-1}$.

Moreover, suppose that $\textrm{dim}\, V\ge 3$. Each component of $\hat{I}$
has as boundary a component of $\hat{E}$ and conversely, every
component of $\hat{E}$ is the boundary of some component of
$\hat{I}$.
\end{proposition}

\begin{proof}
Let $\hat{E}^\alpha$ and $\hat{E}^\beta$ be two components of
$\hat{E}$. They cannot intersect because $\hat{E}$ is embedded. Also
they are finite in number. Indeed, suppose they are not. Then taking
for each $\alpha$ on the index set of components a representative
$p_\alpha\in \hat{E}^\alpha$ there is, by closure of $\hat{E}$ and
the compactness of $S$, some point $p\in \hat{E}$ and an infinite
subsequence $p_{\alpha(n)}$ such that $p_{\alpha(n)}\to p$. But
$p\in \hat{E}^\gamma$ for some $\gamma$, and since $\hat{E}^\beta
\subset \hat{E}$, and the latter is embedded, there is a tubular
neighborhood $\mathcal{T}$ of $\hat{E}^\gamma$ over which $L$
vanishes just on $\hat{E}^\gamma$, in contradiction with
$p_{\alpha(n)}\in \mathcal{T}$, $\alpha(n)\ne \gamma$, for
sufficiently large $n$.

Let us recall that every open bounded convex subset  of
$\mathbb{R}^{n}$ is homeomorphic to the  open ball $D^{n}$ and
has boundary homeomorphic to $S^{n-1}$. The same conclusion holds
for open geodesically convex subsets $C$ of $S^{n}$ such that
$\bar{C}\ne S^{n}$ (use the stereographic projection from some
point in $S^{n}\backslash \bar{C}$ so as to map the convex set to a
bounded convex subset of $\mathbb{R}^{n}$). In particular, the
boundary of $C$ is connected for $n\ge 2$.

By Beem and Perlick result each component $\hat{I}^\gamma$ of
$\hat{I}$  is geodesically convex (on $S$ with the usual metric
induced from the Euclidean space) and $\overline{\hat{I}^\gamma}$
cannot be the whole sphere because $(V\backslash 0)\backslash J\ne \emptyset$.
Thus $\hat{I}^\gamma$ is topologically an open ball $D^{n}$ and
the boundary $\p \hat{I}^\gamma$ is topologically a sphere
$S^{n-1}$, in particular it is connected for $n\ge 2$. The set $\p
\hat{I}^\gamma$ being a subset of the embedded manifold $\hat{E}$
admits a tubular neighborhood on $S$ which does not intersect
$\hat{E}\backslash \p \hat{I}^\gamma$, thus $\hat{E}^\gamma:=\p
\hat{I}^\gamma$ is actually a component of $\hat{E}$. That each
component of $\hat{E}$ is the boundary of some component of
$\hat{I}$ is immediate from $\hat{E}=\p \hat{I}$. $\square$
\end{proof}

\begin{remark}
Taking into account the previous result, the Lemma \ref{juz} can be
easily generalized to each component
\[
\p I^\alpha=\p J^\alpha=E^\alpha, \quad
\overline{I^\alpha}=J^\alpha, \quad \textrm{Int}\,
J^\alpha=I^\alpha,
\]
where we used the topology of $V\backslash 0$ induced from $V$.
\end{remark}

\begin{corollary}
Suppose that $\textrm{dim}\, V\ge 3$. The set $(V\backslash 0)\backslash J$ is a connected open subset of $V\backslash 0$ and analogously,
$\hat{N}:=((V\backslash 0)\backslash J)\cap S$ is a connected open subset of $S$.
\end{corollary}

\begin{proof}
The components
$\hat{J}^\alpha=\overline{\hat{I}^\alpha}$ are convex subsets of $S$
which do not intersect. The boundaries $\p \hat{J}^\alpha$ have topology $S^{n-1}$, $n\ge 2$, and they are not only connected but also path connected \cite{willard70}. Taken any two points $p,q\in \hat{N}$ there is some path contained in $S$  connecting them (since $S$ is path connected). But this path can be easily deformed to pass around  each ball ` $\hat{J}^\alpha$', instead of crossing them, so as to give a path entirely contained in $ \hat{N}$. Thus $\hat{N}$ is path connected and hence connected. $\square$
\end{proof}
%
%
%
%
%
%

We return temporarily to the full manifold picture in order to
answer (v). The result is pretty intuitive but we provide a short proof.

\begin{proposition} \label{nax}
On the full manifold $(M,L)$ the number of components of $I_x$
is independent of $x$.
\end{proposition}

\begin{proof}
Let $N(x_0)$ be the number of components of  $I_{x_0}$. It suffices
to show that this number is constant in an open neighborhood of
$x_0$. Locally near $x_0$ we can treat $TM$ as a product $U\times
\mathbb{R}^{n+1}$, with $U$ neighborhood of $x_0$.

Let $\mathcal{T}\subset S^{n}\subset \mathbb{R}^{n+1}$ be a tubular
neighborhood of $\hat{E}(x_0)$, then $\mathcal{T}$ decomposes into
the disjoint union of tubular neighborhoods $\mathcal{T}^k$ of
$\hat{E}^k(x_0)$, $k=1,\cdots, N(x_0)$. If $x$ is varied the sets
$\hat{E}^k(x)\subset S^{n}$ can change but they remain in the
respective disjoint sets $\mathcal{T}^k$, provided $x$ is
sufficiently close to $x_0$. Thus $N(x)$ does not depend on $x$ in a
neighborhood of $x_0$.

%

\end{proof}

\subsection{The reverse Cauchy-Schwarz and triangle inequalities}

Let us generalize the Cauchy-Schwarz inequality.

As mentioned in the introduction, the proofs will clarify that Eqs. (\ref{nun}) and (\ref{ccl}) also hold in Lorentz-Finsler theories with Lagrangians defined over conic subbundles.

\begin{theorem} [Finslerian reverse Cauchy-Schwarz inequality] Let $v_1,v_2\in
J^\alpha$ then
\begin{equation} \label{nun}
-g_{v_1}({v}_1,{v}_2)\ge \sqrt{-g_{v_1}({v}_1,{v}_1)}\,
\sqrt{-g_{v_2}({v}_2,{v}_2)},
\end{equation}
with equality if and only if $v_1$ and $v_2$ are proportional. In
particular, if $v_1,v_2\in J^\alpha$ then $g_{v_1}({v}_1,{v}_2)\le
0$ with equality  if and only if $v_1$ and $v_2$ are proportional
and lightlike.
\end{theorem}

 This result answers question
(vii) in the negative.
We stress that if $g_{v_1}(v_2,v_2)\le 0$, then a similar inequality holds with $g_{v_2}$ replaced by $g_{v_1}$ in the last term appearing in Eq.\ (\ref{nun}). This is the well known Lorentzian reverse Cauchy-Schwarz inequality \cite{hawking73} for the
metric $g_{v_1}$.

\begin{remark}
By our definition of $J^\alpha$, the theorem does not contemplate the case in which $v_1$ or $v_2$  vanish. Nevertheless, the inequality can be trivially extended by continuity to the case $v_1=0$ or $v_2=0$. It suffices to notice that $g_{a}(a,b)$ is continuous in $(a,b)$, also at $a=0$   provided it is defined to vanish there (Remark \ref{ret}).
\end{remark}


\begin{remark}
The Finslerian reverse Cauchy-Schwarz inequality can be expressed in the following form: for every  $v_1\in I^\alpha$, $v_2\in
J^\alpha$,  we have
\[
v_2^\mu \p_{v_1^\mu} \sqrt{-g_{v_1}(v_1,v_1)}\ge \sqrt{-g_{v_2}(v_2,v_2)}.
\]
with equality if and only if $v_2$ is proportional to $v_1$.
\end{remark}

\begin{proof}
Let $v_1,v_2\in I^\alpha$ and let us
define
\[
\tilde{v}_1=v_1/\sqrt{-2L(v_1)}, \qquad
\tilde{v}_2=v_2/\sqrt{-2L(v_2)},
\]
so that $2L(\tilde{v}_1)=2L(\tilde{v}_2)=-1$, that is $v_1,v_2\in
E^\alpha(1)$. Because of the convexity of $J^\alpha(1)$
\[
2L((1-\alpha)\tilde{v}_1+\alpha \tilde{v}_2)\le -1, \quad \alpha \in
[0,1]
\]
thus
\[
2L(\tilde{v}_1+\frac{\alpha}{{1-\alpha}}\, \tilde{v}_2)\le
-\frac{1}{(1-\alpha)^2}\le -1, \quad \alpha \in [0,1).
\]
Now, let us Taylor expand the $C^3$ function of $\alpha$,
$2L((1-\alpha) \tilde{v}_1+\alpha \tilde{v}_2)$ at $0$
\[
2L((1-\alpha) \tilde{v}_1+\alpha \tilde{v}_2)=-1+\alpha 2
g_{v_1}(\tilde{v}_1,\tilde{v}_2-\tilde{v}_1)+\alpha^2
g_{v_1}(\tilde{v}_2-\tilde{v}_1,\tilde{v}_2-\tilde{v}_1)+o(\alpha^2).
\]
Since, by strict convexity, it is smaller than $-1$ for $\alpha\in (0,1)$, the first order
expansion already implies
\begin{equation} \label{nos}
g_{v_1}(\tilde{v}_1,\tilde{v}_2)\le
g_{v_1}(\tilde{v}_1,\tilde{v}_1).
\end{equation}
Observe that equality holds only if
$g_{v_1}(\tilde{v}_1,\tilde{v}_2-\tilde{v}_1)=0$ namely
$\tilde{v}_2-\tilde{v}_1$ is $g_{v_1}$-orthogonal to the
$g_{v_1}$-timelike vector $\tilde{v}_1$, hence
$\tilde{v}_2-\tilde{v}_1$ is $g_{v_1}$-spacelike or zero. But again,
 the Taylor expansion at the second order shows that
$\tilde{v}_2-\tilde{v}_1=0$ thus $v_1$ and $v_2$ are proportional.

Equation (\ref{nos}) reads  equivalently: for $v_1,v_2\in I^\alpha$
\[
-g_{v_1}({v}_1,{v}_2)\ge \sqrt{-g_{v_1}({v}_1,{v}_1)}\,
\sqrt{-g_{v_2}({v}_2,{v}_2)},
\]
with equality if and only if $v_1$ and $v_2$ are proportional.

By continuity this inequality holds for $v_1,v_2\in
\overline{I^\alpha}=J^\alpha$. Clearly, if $v_1$ and $v_2$ are proportional the
inequality becomes an equality. In order to study the equality case,
we can assume that $v_1$ or $v_2$ is lightlike, since the case in
which they are both timelike has been already shown above to lead to
the proportionality of these vectors. Thus suppose that $v_1$ is
lightlike, and let us Taylor expand the $C^3$ function of $\alpha$,
$2L((1-\alpha) {v}_1+\alpha v_2)$ at $0$
\[
0\ge 2L((1-\alpha){v}_1+\alpha {v}_2)=\alpha 2
g_{v_1}({v}_1,{v}_2-{v}_1)+\alpha^2
g_{v_1}({v}_2-{v}_1,{v}_2-{v}_1)+o(\alpha^2),
\]
where the first inequality is due to the convexity of  $J^\alpha$.
Since we are assuming equality in Eq.\ (\ref{nun})
$g_{v_1}({v}_1,{v}_2)=0$, thus $v_2$ is either proportional to the
lightlike vector $v_1$ or it is $g_{v_1}$-spacelike. Furthermore,
the previous inequality simplifies to
\[
0\ge  \alpha^2 g_{v_1}({v}_2,{v}_2)+o(\alpha^2),
\]
thus $v_2$ cannot be $g_{v_1}$-spacelike and hence must be
proportional to $v_1$ as we wished to prove. $\square$
\end{proof}

We are now able to answer  question (v).

\begin{corollary} \label{sha}
If $v_1,v_2\in J^\alpha$ and $v_1=a v_2$, $a\in \mathbb{R}$, then
$a> 0$. In other words, $J^\alpha\cup \{0\}$ is a convex cone which does not
contain any line passing through the origin (sharp cone).
\end{corollary}

\begin{proof}
Let us take $v\in J^\alpha$, in such a way that $v\ne v_1,v_2$ and
$v\in I^\alpha$, then by the Finslerian reverse Cauchy-Schwarz
inequality $g_v(v,v_1), g_v(v,v_2)<0$, from which the desired
conclusion follows. $\square$
\end{proof}

\begin{theorem} [Finslerian reverse triangle inequality]
Let $v_1,v_2\in J^\alpha$ then defined $v=v_1+v_2$, we have $v\in
J^\alpha$ and
\begin{equation} \label{ccl}
\sqrt{-g_{v}(v,v)}\ge
\sqrt{-g_{v_1}({v}_1,{v}_1)}+\sqrt{-g_{v_2}({v}_2,{v}_2)}.
\end{equation}
with equality if and only if $v_1$ and $v_2$ are proportional. In
particular,  if $v_1$ is timelike and $v_2$ is causal then $v$ is
timelike.
\end{theorem}

\begin{proof}

Let $v_1,v_2\in J^\alpha$,  since $J^\alpha$ is a convex cone $v\in
J^\alpha$. If $v_1$ and $v_2$ are proportional and lightlike then the
triangle inequality is obvious and it is actually an equality. If
$v_1$ and $v_2$ are not proportional to the same lightlike vector,
then by Theorem \ref{jod} $v\in I^\alpha$. In this case let us divide
the identity
\begin{equation} \label{dol}
-g_{v}(v,v)=-g_{v}(v,v_1)-g_{v}(v,v_2) ,
\end{equation}
by $\sqrt{-g_{v}(v,v)}$. Using Eq.\ (\ref{nun})
\begin{align*}
\sqrt{-g_{v}(v,v)}&=-\frac{g_{v}({v},{v}_1)}{\sqrt{-g_{v}(v,v)}}-\frac{g_{v}({v},{v}_2)}{\sqrt{-g_{v}(v,v)}}\\
&\ge \sqrt{-g_{v_1}({v}_1,{v}_1)}+\sqrt{-g_{v_2}({v}_2,{v}_2)} .
\end{align*}
Observe that if equality holds then the reverse Cauchy-Schwarz
inequalities which we used should hold with the equality sign, thus
$v_1$ and $v_2$ are both proportional to $v$ and hence among
themselves. $\square$
\end{proof}

\subsection{The Legendre map}

\begin{definition}
The $C^{2,1}$ map $\ell\colon V\backslash0 \to
V^*\backslash0$ given by
\[v \to g_v(v,\cdot)=\p L/\p
v,\] is called {\em Legendre map}. The extension $\ell\colon V \to
V^*$ obtained defining $\ell(0):=0$ is given the same name.
\end{definition}

We remark that the extension is Lipschitz at the origin because
$g_v$ is bounded on the unit sphere $S\subset V$ and hence, by
homogeneity, all over the unit ball (Remark \ref{ret} and \cite[Rem.\ 1.3]{minguzzi13d}).

We start giving a partial affirmative answer to question (viii) in any dimension.

\begin{theorem}[On-shell injectivity of the Legendre map]
\label{aou}
 Let $v_1,v_2\in J^\alpha$ and suppose that
\[
\ell(v_1)=\ell(v_2)
\]
then $v_1=v_2$.
\end{theorem}

\begin{proof}
Let $v_1,v_2\in J^\alpha$ be such that
$\ell(v_1)=\ell(v_2)$.
We can assume
$v_1,v_2\ne 0$ otherwise the conclusion is obvious from the
non-degeneracy of $g_v$. If $v_1$ is lightlike then $g_{v_2}(v_2,v_1)=0$ thus by Theorem \ref{nun} and Cor.\ \ref{sha}, $v_2=av_1$, $a>0$. Thus $\ell(v_1)=\ell(v_2)$ and the non-degeneracy of $g$ imply $a=1$, and finally $v_2=v_1$.

So we can assume $v_1,v_2\in I^\alpha$. Using the reverse Cauchy-Schwarz inequality
we obtain
\begin{equation} \label{huj}
-g_{v_1}({v}_1,v_1)=-g_{v_2}({v}_2,v_1)\ge
\sqrt{-g_{v_1}({v}_1,v_1)} \,\sqrt{-g_{v_2}({v}_2,v_2)},
\end{equation}
thus $\sqrt{-g_{v_1}({v}_1,v_1)}\ge  \sqrt{-g_{v_2}({v}_2,v_2)}$ and
exchanging the roles of $v_1$ and $v_2$,
\begin{equation}  \label{suh}
{g_{v_1}({v}_1,v_1)}=  {g_{v_2}({v}_2,v_2)}.
\end{equation}
Thus
\[
-g_{v_2}({v}_2,v_1)=-g_{v_1}({v}_1,v_1)=\sqrt{-g_{v_2}({v}_2,v_2)}\,
\sqrt{-g_{v_1}({v}_1,v_1)}
\]
which shows that the reverse Cauchy-Schwarz inequality used in Eq.\
(\ref{huj}) holds with the equality sign thus $v_1\propto v_2$,
which by  Eq.\ (\ref{huj}) implies $v_1=\pm v_2$ and by Cor.\
\ref{sha} $v_1=v_2$. $\square$
\end{proof}

We are actually  able to answer affirmatively question (viii).

\begin{theorem}[The Legendre map is a diffeomorphism for $n\ge 2$ ] {}\\
\label{inj}
 Suppose that $\textrm{dim}\, V\ge 3$, then the map
$\ell\colon V\backslash0 \to V^*\backslash0$ defined by
\[
v \mapsto g_v(v,\cdot)=\p L/\p v
\]
is a bijection and hence a $C^{2,1}$-diffeomorphism. Its extension
to the whole $V$ is a Lipeomorphism (locally Lipschitz homeomorphism with locally Lipschitz inverse).
\end{theorem}

\begin{remark} \label{mdl}
This result holds with $(V,L)$ general pseudo-Minkowski space
because in the proof we do not use the signature of $g_v$.
\end{remark}

\begin{proof}
Let us introduce the equivalence relation on $V\backslash 0$ ``$v_1\sim v_2$ if
there is $s>0$ such that $v_1=s v_2$", namely let us regard $V\backslash 0$ as a
radial bundle over a base $Q$ diffeomorphic to $S^{n}$.
Analogously, let us introduce the quotient  $Q^*$ of $V^*\backslash 0$ with
respect to the radial directions. The map $\ell$ satisfies
$\ell(sv)=s\ell(v)$ for every $s>0$, thus it passes to the quotient
to a map $\tilde\ell\colon Q \to Q^*$.

The map $\ell$ is a local diffeomorphism because the Jacobian $\p
\ell/\p v=\textrm{Hess} L=g_v$ is non-singular. As a consequence,
$\tilde\ell$ has a Jacobian having maximal rank (i.e. $n$), thus
it is also a local diffeomorphism. Since $Q$ and $Q^*$ are closed
manifolds with the same dimension, $\tilde\ell$ is actually a
covering \cite{dubrovin85}. Since $Q^*\sim S^{n}$ is simply
connected (here we use $n\ge 2$) and $Q\sim S^{n}$ is connected
this covering is actually a homeomorphism. Thus if $v_1,v_2$ are
such that $\ell(v_1)=\ell(v_2)$ we have
$\tilde{\ell}([v_1])=\tilde\ell([v_2])$ then $[v_1]=[v_2]$ and hence
$v_1=s v_2$, and from $\ell(v_2)=\ell(v_1)=s \ell(v_2)$ we have that
$s=1$, that is $v_1=v_2$, namely $\ell$ is injective. Chosen $p\in V^*$
there is $v\in V$ such that $\tilde\ell([v])=[p]$, namely
$\ell(v)=sp$ for some $s>0$, but then $\ell(v/s)=p$, thus $\ell$ is
surjective. $\square$
\end{proof}


\begin{definition}
A hyperplane  $W\subset V$ passing through the origin is $\alpha$-{\em spacelike}
($\alpha$-{\em null}) if $W\cap J^\alpha=\emptyset$  (resp.\ $W\cap
I^\alpha=\emptyset$ and $W\cap J^\alpha\ne \emptyset$). It is $\alpha$-timelike if it is neither
$\alpha$-spacelike nor $\alpha$-null.
\end{definition}

\begin{proposition} \label{bbg}
Let $W$ be a $\alpha$-spacelike hyperplane, then up to a positive
proportionality constant there is one and only one $u\in J^\alpha$
such that
\begin{equation} \label{aop}
W=\{v\in V\colon g_{u}(u,v)=0\},
\end{equation}
and this vector can be normalized so that $g_u(u,u)=-1$.

Let $W$ be a $\alpha$-null hyperplane, then up to a positive
proportionality constant there is one and only one $u\in J^\alpha$
such that Eq.\ (\ref{aop}) holds. This vector belongs to $E^\alpha$.

Conversely, for every $u\in J^\alpha$ the hyperplane given by Eq.\
(\ref{aop}) is  $\alpha$-spacelike or $\alpha$-null (by the previous
statement, spacelike iff $u$ is timelike or null iff $u$ is lightlike).
\end{proposition}

Observe that since $J^\alpha$ is sharp (Cor.\ \ref{sha}) there exist
$\alpha$-spacelike hyperplanes.

\begin{proof}
Suppose that $W$ is $\alpha$-spacelike. Let $v\in J^\alpha(1)$, and
let $\{b_i\}$, $b_1=v$, $b_i\in W$ for $i\ge 2$, be a base of $V$.
Let $\{v^i\}$ be the corresponding coordinates. The set
$J^\alpha(1)\cap \{v: v^1\le 1\}$ is compact  thus there is a
minimum value $s_0$ of $v^1$ over it. Thus $s_0$ is the minimum
value of $s\ge 0$ for which $(s v+W)\cap J^\alpha(1)\ne \emptyset$.
Necessarily, $s_0>0$. Let $u\in (s_0 v+W)\cap J^\alpha(1)$ then the
hyperplane $u+W$ is tangent to the strictly convex set $J^\alpha(1)$
at $u$, that is for every $w\in W$, $0=\dd L(u)(w)=g_{u}(u,w)$.
Since $g_u(u,u)=-1$, we have found the searched $u$.

Suppose that $u'\in J^\alpha(1)$ is another vector with the same
property, then $u'+W$ is tangent to $J^\alpha(1)$ at $u'$, which
contradicts the strict convexity of $J^\alpha(1)$ (no two parallel
hyperplanes can be tangent to a strictly convex set leaving the convex set on the same side of the hyperplane).

Suppose that $W$ is $\alpha$-null, then $W\cap E^\alpha\ne
\emptyset$. Let $u\in W\cap E^\alpha$, then $W$ is tangent to the
convex set  $J^\alpha$ at $u$, thus for every $w\in W$, $0=\dd
L(u)(w)=g_{u}(u,w)$. Since $u\in E^\alpha$, $g_u(u,u)=0$. Suppose
that there is another vector $u'\in J^\alpha$ such that for every
$w\in W$, $g_{u'}(u',w)=0$, then since $u\in W$, $g_{u'}(u',u)=0$.
By the equality case in the Finsler reverse Cauchy-Schwarz
inequality it must be $u'=s u$ for some $s\ge 0$.

Finally, let $u\in J^\alpha$ and let $W$ be defined through Eq.\
(\ref{aop}). Suppose that $W$ intersects $J^\alpha$, namely suppose that $W$ is not
$\alpha$-spacelike, and let $w\in W\cap J^\alpha$ then
\[
0=g_u(u,w)\le -\sqrt{-g_u(u,u)}\sqrt{-g_w(w,w)},
\]
 where we used the
Finsler reverse Cauchy-Schwarz inequality. Thus equality must hold which implies that $w$ and $u$ are both proportional and lightlike. Thus $W$ can intersect $J^\alpha$ only on $E^\alpha$, namely  $W$ is $\alpha$-null. In conclusion, if $W$ is defined through Eq.\ (\ref{aop}) with $u\in J^\alpha$
then it is $\alpha$-spacelike or $\alpha$-null. $\square$
\end{proof}

\begin{definition}
Given an $\alpha$-spacelike hyperplane $W$, the vector $u$ selected
by  Prop.\ \ref{bbg} is called {\em $\alpha$-normal} to $W$.
\end{definition}

Let us recall that an half-space is a closed subset of $V$ which
lies on one side of a hyperplane $W\subset V$ passing through the origin.
A {\em cone} $C$ is a subset of $V$ which is closed under  product by a positive scalar. A proper closed convex cone containing the origin is the intersection of the half-spaces containing the cone. Conversely, any such intersection is a proper closed convex cone containing the origin. The convex envelop of a union of convex cones is a convex cone.

\begin{lemma} \label{hfk}
Let $\textrm{dim}\, V\ge 3$. No two components of $J$ can lie on the
same half-space of $V$. Equivalently, the convex envelop of any two distinct
components $J^\alpha,J^\beta\subset J$ is the whole $V$:
$\textrm{Conv}(J^\alpha \cup J^\beta)=V$. Another equivalent
formulation of this property is $\textrm{Conv}(I^\alpha \cup
I^\beta)=V$.
\end{lemma}

\begin{proof}
Let us show that the first two claims are equivalent. Indeed, the
latter implies the former because the convex envelop of
$J^\alpha$, $J^\beta$, at one side $V^+$ of $W$ cannot contain the
other side $V^-\backslash W$. Conversely, the former claim implies
the latter because $J^\alpha \cup J^\beta$ is not contained in any
half-space.

Suppose that $n\ge 2$, so that $\ell$ is injective. Suppose by
contradiction that there is a hyperplane $W$ and two components
$J^\alpha$, $J^\beta$, which lie at the same side of $W$. Then by
Prop.\ \ref{bbg} there are $u^\alpha \in J^\alpha$,
$u^\beta\in J^\beta$ such that defined $p_\alpha=
g_{u^\alpha}(u^\alpha,\cdot)$ and
$p_\beta=g_{u^\beta}(u^\beta,\cdot)$
\[
W=\ker p_\alpha=\ker p_\beta.
\]
In order to prove that $p_\alpha=s p_\beta$ for some $s>0$ it
suffices to show that both $p_\alpha$ and $p_\beta$ take negative
value at some point on the side of $W$ which contains the two cones.
This is true because taken $v^\alpha\in I^\alpha$ and $v^\beta\in
I^\beta$, $v^\alpha,v^\beta$ lie on the same side of $W$ and by the
Finsler reverse Cauchy-Schwarz inequality $p_\alpha(v^\alpha)<0$,
$p_\beta(v^\beta)<0$.
Thus $g_{u^\alpha}(u^\alpha,\cdot)=s
g_{u^\beta}(u^\beta,\cdot)=g_{su^\beta}(su^\beta,\cdot)$ namely
$\ell$ is not injective, a contradiction.

For the last claim, suppose that $\textrm{Conv}(J^\alpha \cup
J^\beta)=V$. Let us recall that for every set $S$, $\textrm{Conv}\,
\overline{S}\subset \overline{\textrm{Conv}\, S}$, thus
\[
V=\textrm{Conv}( J^\alpha \cup
J^\beta)=\textrm{Conv}(\overline{I^\alpha \cup I^\beta})\subset
\overline{\textrm{Conv}(I^\alpha \cup I^\beta)}
\]
It is well known that for every open and convex set $A\subset \mathbb{R}^k$ we have
$A=\textrm{Int}\overline{A}$ (i.e.\ $A$ is regular open) thus with
$A=\textrm{Conv}(I^\alpha \cup I^\beta)$ we conclude that $A=V$. $\square$
\end{proof}

\begin{proposition} \label{ksg}
Let $\textrm{dim}\, V\ge 3$. Given any two distinct components
$I^\alpha$ and $I^\beta$  we have $I^\alpha \cap (-I^\beta)\ne
\emptyset$.
\end{proposition}

\begin{proof}
We have $\textrm{Conv}(I^\alpha \cup I^\beta)=V$, thus taken $v\in
I^\beta$ there are $v_1\in I^\alpha$ and $v_2\in I^\beta$ such that
$-v=a v_1+b v_2$ with $a+b=1$, $a,b\ge 0$, and where $a>0$ because
$I^\beta\cap(-I^\beta)=\emptyset$ as $I^\beta$ is sharp. Thus
$av_1\in (-I^\beta)\cap I^\alpha$. $\square$
\end{proof}

We are now able to prove the main theorem of this work. It answers question (vi) in the reversible case.

\begin{theorem}
Let $\textrm{dim}\, V\ge 3$ and suppose that $L$ is reversible, then $I$ has two components.
\end{theorem}

\begin{proof}
Let $I^\alpha$ be a component of $I$, and let $I^{-\alpha}=-I^{\alpha}$ be its opposite. Suppose that there is another component $I^\beta$ with $\beta\ne \alpha,-\alpha$, then  $I^\beta$ cannot intersect neither $I^\alpha$ nor $I^{-\alpha}$ since the components of $I$ are disjoint, however it must intersect both of them by  Prop.\ \ref{ksg}. The contradiction proves that there are just two components. $\square$
%
\end{proof}

The following result strengthens the Finslerian reverse Cauchy-Schwarz inequality.

\begin{corollary}
Let $\textrm{dim}\, V\ge 3$ and suppose that $L$ is reversible, then for every $v\in I$, and $w\in V$ such that $g_v(v,w)=0$ we have $L(v+w)\ge L(v)$, equality holding if and only if $w=0$.
\end{corollary}

\begin{proof}
Let $J^\alpha(c)$ be the component of $J(c)$ such that $v\in J^\alpha$, $g_v(v,v)=-c^2$. Since $J^\alpha(c)$ is strictly convex $v+w$ belongs to it iff $w=0$. Observe that $\ker g_v(v,\cdot)$ is $\alpha$-spacelike, hence by reflexivity, $-\alpha$-spacelike. The hyperplane $v+\ker g_v(v,\cdot)$ cannot intersect $I^{-\alpha}$ and so, since there are only two components, does not intersect any other component of $I$ but $I^\alpha$ at $v$. Thus the desired conclusion follows since $L$ is positive outside $I^\alpha$. $\square$
\end{proof}

\begin{proposition} \label{ofa}
Let $v\in V$ and suppose that for every $u\in E^\alpha$, $g_u(u,v)\le 0$, then $v\in J^\alpha\cup\{0\}$.
\end{proposition}

The hypothesis is implied by the alternative assumption: $g_u(u,v)\le 0$, for every  $u\in I^\alpha$, using continuity.

\begin{proof}
The condition $g_u(u,v)\le 0$, states that $v$ is contained in an half-space $H_u$ which is bounded by the $\alpha$-null tangent hyperplane to $J^\alpha$, $\ker g_u(u,\cdot)$, and which contains $J^\alpha$. But since any proper closed convex cone containing the origin is the intersection of the half spaces that contain it, $v$ is actually in the intersection $\cap_u H_u$ which coincides with $J^\alpha\cup\{0\}$. $\square$
\end{proof}

\subsection{Polar cones and the Hamiltonian}

Let $V^*$ be the  space dual to $V$.
 The polar cone to $J^\alpha$ is
\[
J^{\alpha*}=\{p\in V^*\backslash 0: p(v)\le 0 \ for \ every \ v\in J^\alpha \}.
\]
General properties of a polar cone are: it is convex, closed in $V^*\backslash 0$, its
polar cone is the original closed cone $J^\alpha$, and since $J^\alpha$ is
sharp and with non-empty interior so is $J^{\alpha*}$.

\begin{lemma} \label{cch}
If $p\in J^{\alpha*}$ then there is $u \in J^\alpha$ such that
\[
p=g_u(u,\cdot).
\]
\end{lemma}

\begin{proof}
If $p\in J^{\alpha*}$, $p\ne 0$, then $\ker p$ can
intersect $J^\alpha$ only on its boundary $E^\alpha$, indeed if
$u\in I^\alpha\cap \ker p$ then $0=p(u)$ and by continuity and
linearity we can find $u'\in J^\alpha$ in a neighborhood of $u$ such
that $p(u')>0$, a contradiction. As a consequence, $W:=\ker p$ is
either $\alpha$-spacelike or $\alpha$-null (Prop.\ \ref{bbg}), that
is there is $u$ $\alpha$-timelike or $\alpha$-lightlike such that $\ker
p=\ker g_u(u, \cdot)$. Observe that $p(u)\le 0$ by definition of polar cone and $g_u(u,u)\le 0$ since $u\in J^\alpha$, thus $u$ can be rescaled with a positive constant to obtain the required equation. $\square$
\end{proof}


\begin{corollary}
The Legendre map is a bijection between $J^\alpha$ and
$J^{\alpha*}$ (thus a homeomorphism).
\end{corollary}

\begin{proof}
By Theorem \ref{aou} we already known that $\ell$ is
injective. From Lemma \ref{cch} we have that $J^{\alpha*}$ is in the image of $J^\alpha$, and from the Finslerian reverse Cauchy-Schwarz inequality we have that the image of  $J^\alpha$ is in  $J^{\alpha*}$. $\square$
\end{proof}

As a consequence we may call $p \in J^{\alpha*}$, $\alpha$-timelike or $\alpha$-lightlike depending on the corresponding character of $u \in J^\alpha$, or say that $p$ is $\alpha$-spacelike if $p\notin J^{\alpha*}$. The previous result then implies that $I^{\alpha*}$ so defined is  regular open since $I^\alpha$ is, and that $J^{\alpha*}=\overline{I^{\alpha*}}$.
Lemma \ref{hfk} leads to

\begin{proposition}
Let $\textrm{dim} V\ge 3$. The polar cones $J^{\alpha*}$  do not intersect.
\end{proposition}

\begin{proof}
Indeed, if $p\in J^{\alpha*}\cap J^{\beta*}$ then there is an $\alpha$-causal vector such that $p=g_u(u,\cdot)$ and a  $\beta$-causal vector  such that $p=g_v(v,\cdot)$ in contradiction with $J^\alpha\cap J^\beta=\emptyset$ and the injectivity of the Legendre map.  $\square$
\end{proof}

\begin{proposition}
Let $\textrm{dim}\, V\ge 3$. No two distinct components of $J^*$ can lie on the
same half-space of $V^*$. Equivalently, the convex envelop of any two
components $J^{\alpha*},J^{\beta*}$ is the whole $V^*$:
$\textrm{Conv}(J^{\alpha*} \cup J^{\beta*})=V^*$.
\end{proposition}

\begin{proof}
Suppose by contradiction that $J^{\alpha*},J^{\beta*}$ lie in the same half-space of $V^*$. This means that there is $v\in V\backslash0$ such that for every $p\in J^{\alpha*}$, and every $q\in J^{\beta*}$, $p(v),q(v)\le 0$. Stated in another way, for every $u\in J^\alpha$, $g_u(u,v)\le 0$, which implies $v\in J^\alpha$ by Prop.\ \ref{ofa}. Arguing similarly, we obtain $v\in J^\beta$, a contradiction since $J^\alpha$ and $J^\beta$ do not intersect. The proof of the last claim is analogous to the last one of Lemma \ref{hfk}. $\square$
\end{proof}

\begin{proposition}
Let $\textrm{dim}\, V\ge 3$. Given any two distinct components
$I^{\alpha*}$ and $I^{\beta*}$  we have $I^{\alpha*} \cap (-I^{\beta*})\ne
\emptyset$.
\end{proposition}

\begin{proof}
The proof is analogous to that of Prop.\ \ref{ksg}. Alternatively, this result is simply the statement that $I^\alpha$ and $I^\beta$ can be separated by a hyperplane, which is a consequence of the convexity of $\hat{I}^\alpha$ and $\hat{I}^\beta$. $\square$
\end{proof}

These results clarify that the polar cones satisfy properties  analogous to those satisfied by the original cones.

 Let
$\textrm{dim} V\ge 3$ and let us denote with $g_p$ the inverse of $g_v$ where $v$ is such that $\ell(v)=p$ and let
\begin{equation} \label{hem}
H(p)=\frac{1}{2}g_p(p,p),
\end{equation}
so that $H(p)=L(v)$.
This function is positive homogeneous of second degree in $p$ (since the matrix of $g_v$ and hence that of $g_p$ have zero degree), and its Hessian is
\begin{align*}
\frac{\p^2 H}{\p p_\alpha \p p_\beta}&=\frac{\p}{\p p_\alpha}[ \frac{\p H}{\p v^\gamma} \frac{\p v^\gamma}{\p p_\beta}]= \frac{\p}{\p p_\alpha}[ \frac{\p L}{\p v^\gamma} (\frac{\p p_\beta}{\p v^\gamma})^{-1}]=\frac{\p}{\p p_\alpha}[ p_\gamma (g_{v}^{-1})_{ \beta \gamma}]\\
&= (g_{v}^{-1})_{ \beta \alpha}+p_\gamma (\frac{\p v^\mu}{\p p_\alpha}) \frac{\p}{\p v^\mu} (g_{v}^{-1})_{ \beta \gamma}=g_p-p_\gamma (\frac{\p v^\mu}{\p p_\alpha}) \frac{\p}{\p v^\mu} (g_{v}^{-1})_{ \beta \gamma}=g_p,
\end{align*}
where in the last step we used Eq.\ (\ref{koi})
\[
p_\gamma \frac{\p}{\p v^\mu} (g_{v}^{-1})_{ \beta \gamma}=-  (g_{v}^{-1})_{ \beta \delta}  g_{v\, \delta \nu, \mu}  (g_{v}^{-1})_{ \nu \gamma} p_\gamma= -  (g_{v}^{-1})_{ \beta \delta}  g_{v\, \delta \nu, \mu} v^\nu=0.
\]
As $g_p$ has the same signature of $g_v$ we conclude that

\begin{proposition}
Let $\textrm{dim} V\ge 3$. Let $g_p$, $p=\ell(v)$, be the inverse of $g_v$ and let $H\colon V^* \to \mathbb{R}$ be defined through Eq.\ (\ref{hem}). Then  $(V^*, H)$ is a pseudo-Minkowski space, $g_p$ is the Hessian of $H$ and $\ell^*H=L$.
\end{proposition}

The next result shows that $F^*:=\vert 2H\vert^{1/2}$ is a sort of norm dual to $F:=\vert 2L\vert^{1/2}$, once we restrict ourselves to the timelike domains.

\begin{proposition}
Let $\textrm{dim} V\ge 3$. We have
\[
 \forall p\in  I^{\alpha*}, \quad F^*(p)=\inf_{v\in I^\alpha:  F(v)=1} \vert p(v)\vert.
\]
Under reversibility we can also replace $I^{\alpha*}$ with $I^{*}$ and $I^\alpha$ with $I$.
\end{proposition}

\begin{proof}
There is $w\in  I^\alpha$ such that $p=g_w(w,\cdot)$, thus the first claim is a trivial consequence of the Finslerian reverse Cauchy-Schwarz inequality.

Let us look more in detail at the proof of the strengthened version for reversible $L$.
There is $w\in I$ such that $p=g_w(w,\cdot)$. Let us call $I^\alpha$ the component of $I$ which contains $w$. Since there are just two components, if $2L(v)=-1$ we have $v\in J^\alpha$ or $-v\in J^\alpha$.
By the reverse Cauchy-Schwarz inequality applied to either $(w,v)$ or $(w,-v)$, we have $\vert p(v)\vert^2\ge \vert 2L(v)\vert \, \vert 2L(w) \vert= \vert 2H(p)\vert$, where in the last step we used $\ell^*H=L$. The infimum is attained, just take $v=w/{\vert 2L(w)\vert}^{1/2}$, so that $\vert p(v)\vert ={\vert 2L(w)\vert}^{1/2}= {\vert 2H(p)\vert}^{1/2}$. $\square$
\end{proof}

\subsection{Implications of Borsuk-Ulam's theorem for reversibility}

The Borsuk-Ulam theorem states that every continuous function from  $S^n$ into $R^n$ maps some pair of antipodal points to the same point.
Let $n\ge 2$. For positive $s$, the Legendre map satisfies $\ell(sv)=s\ell(v)$ thus it sends the homothety-quotient sphere $S^{n}$ obtained from $V\backslash0$ to the analogous sphere obtained from $V^*\backslash0$. Choose a vector $w\in V^\backslash0$ and let $\pi_w: V\to Q_w$ be the projection of $V$ onto the space $Q_w$ of equivalence classes of vectors in $V$, where any two vectors on the same equivalence class differ by a term proportional to $w$. We introduce a Euclidean scalar product on $V^*$, for instance the canonical one induced from a choice of affine coordinates, so as to identify the homothety-quotient sphere with the actual Euclidean sphere $S^*$.
The next result follows from the Borsuk-Ulam theorem applied to $\pi_w\circ \ell^{-1}\vert_{S^*}$.

\begin{proposition}
Let $\textrm{dim} V\ge 3$. For every $w\in V\backslash0$ we can find $v_1,v_2\in V\backslash0$, such that $g_{v_2}(v_2,\cdot)=-g_{v_1}(v_1,\cdot)$ and $v_2-v_1=2 w$.
\end{proposition}

\begin{proof}
By the  Borsuk-Ulam theorem applied to $\pi_w\circ \ell^{-1}\vert_{S^*}$ there are $p_1,p_2\in S^*$, $p_2=-p_1$, such that $\ell^{-1}(p_2)-\ell^{-1}(p_1)=2\beta w$, where $\beta$ is some constant. The constant can be chosen non-negative inverting the roles of $p_1$ and $p_2$, if necessary. Moreover, $\beta$ cannot vanish for otherwise $\ell^{-1}(-p_1)=\ell^{-1}(p_1)$ in contradiction with the injectivity of $\ell^{-1}$.  Defining $v_1=\ell^{-1}(p_1)/\beta $, $v_2=\ell^{-1}(p_2)/\beta$ gives the desired conclusion. $\square$
\end{proof}

Under reversibility this proposition is trivial, just take $v_2=-v_1=w$.
Thus the Borsuk-Ulam theorem leads us to a result which mitigates the lack of reversibility.

We can also take $q\in V^*\backslash0$ and consider the projection $\pi_q: V^*\to Q_q$, where two elements of $V^*$ belong to $Q_q$ if they differ by a term proportional to $q$. Let us introduce a Euclidean scalar product on $V$ so as to identify the homothety-quotient sphere with the Euclidean sphere $S$. We can then apply the Borsuk-Ulam theorem to $\pi_q\circ \ell\vert_S$.

\begin{proposition}
Let $\textrm{dim} V\ge 3$. The map $\phi\colon  V\backslash0 \to  V^*\backslash0$, defined by $\phi(v)=[g_{v}(v,\cdot)+g_{-v}(v,\cdot)]/2$ is surjective.
\end{proposition}

\begin{proof}
Let $q\in V^*\backslash0$, by the Borsuk-Ulam theorem there are $v_1,v_2\in S$, $v_2=-v_1$ such that $\ell(v_1)-\ell(v_2)=2\beta q$, or equivalently $g_{v_1}(v_1,\cdot)-g_{v_2}(v_2,\cdot)=2\beta q$.  Here $\beta$ is some constant which can be chosen non-negative reversing the roles of $v_1$ and $v_2$, if necessary. Actually, $\beta$ does not vanish, for otherwise $\ell(v_1)=\ell(-v_1)$ in contradiction with the injectivity of $\ell$.  Defining $v=v_1/\beta$ gives the desired conclusion. $\square$
\end{proof}

Once again, under reversibility this claim is trivial, as $\ell$ is surjective. The next result is also interesting.

\begin{proposition}
There is $v\in V \backslash 0$ such that $\ell(-v)=-s \ell(v)$ where $s>0$ is some constant.
\end{proposition}

\begin{proof}
Let us consider the map $w\mapsto -\ell^{-1}(-\ell(w))$. Since it is positive homogeneous of degree 1 it establishes an homeomorphism between the homothety-quotient $S^{n}$ sphere and itself. Thus it has a fixed point $v$,  that is $-\ell^{-1}(-\ell(v))=k v$ with $k>0$, or $k \ell(-v)=-\ell(v)$. $\square$
\end{proof}

\begin{remark}
The results of this section hold with $(V,L)$ general pseudo-Minkowski space
because in the proof we did not use the signature of $g_v$ (see remark \ref{mdl}).
\end{remark}

\section{Conclusions}

Many authors proposed  Finslerian generalizations of Einstein's general relativity (see  \cite{takano74,ishikawa80,asanov85,rutz93,miron92,chang08,vacaru10,pfeifer11,pfeifer12} and their bibliographies). Unfortunately, there is no consensus on the set  of equations as different authors gave different generalizations. In spite of this, Finslerian extensions of Einstein's relativity theory have been considerably popular and found many applications, see \cite{gibbons07,chang08,kouretsis09,li10,chang13}, just to mention a few recent works.
Often,
the dynamical equations are obtained through a tensorial analogy with Einstein's equations, but so far no proof has been given that the new set of equations implies a conservation of the stress-energy tensor or   the geodesic principle (weak equivalence principle) as in general relativity, not to mention the symmetry of the stress energy tensor. Thus, the tensorial equations have non-transparent physical content.  Given this situation it seems advisable to take a step back, avoid the dynamical equations, e.g.\ as in \cite{lammerzahl12}, and try to understand the physics behind the formalism of Finsler geometry. This work goes in this direction trying to elucidate the local causal structure of these theories. Concerning the issue of local causality there have been two approaches.

In Asanov's approach one restricts the Finsler Lagrangian to a conic subbundle of the tangent space. In this way one essentially selects by hand the physically meaningful future timelike cone. The causal structure at a point becomes  trivially the desired one, but there seems to be a price to be paid. The Finsler Lagrangian can often be continued beyond the light cone, but the extension is ignored and discarded as unphysical. Moreover, the description of the domain of the Finsler Lagrangian might require the introduction of additional fields, and finally, local spacelike geodesics might not make sense at all, a fact which might cause several interpretational problems in connection with ideal measurements of length.

In Beem's approach the Finsler Lagrangian is well defined over the whole slit tangent bundle. Here it makes sense to consider the problem of the number of light cones predicted by the theory.
 In the early seventies Beem  showed that at each point the Finsler metric can determine arbitrarily many light cones on the tangent space. This result suggested that Finsler geometry could be geometrically unsuitable for developing generalizations of general relativity. Indeed,
 since then no progress was  made in Beem's direction.

In this work we have established that, actually, every reversible Finsler spacetime has the usual light cone structure at any point provided the spacetime dimension is larger than two. That is, at each point we have just two strictly convex sharp  causal cones intersecting  at the origin whose linear span is the whole tangent space.

This result proves that the pathologies of the two dimensional case do not show up in more dimensions, and clarifies that Finsler spacetimes are indeed geometrically interesting generalizations of Lorentzian spacetimes with possibly deep physical implications.  In fact, in this theory there is no need to restrict the Finsler Lagrangian to some cone subbundle,  nor to introduce additional dynamical fields to describe those subbundles.
Thus, this work on local causality joins and confirms  this author's conclusions  \cite{minguzzi13d}
that most results of global causality theory extend from the Lorentzian to the Finslerian domain, provided the local theory leaves us with one future cone (either selected `by hand' as in Asanov's approach or deduced from our theorem as in Beem's approach).
As recalled in a recent paper  \cite{javaloyes13} another approach to global causality might follow  Fathi and Siconolfi \cite{fathi12} work on cones structures.

Finally, Beem's approach has a considerable advantage over Asanov's. Namely, the family of allowed Finsler Lagrangians is much smaller, so it can provide strong hints at the dynamics helping to select the correct generalization of Einstein's equations.

In a next work we shall consider the dynamical equations and we shall provide further results for non-reversible metrics.

\section*{Acknowledgments} I thank Marco Spadini and an anonymous referee for their critical reading of the manuscript. I thank Christian Pfeifer for pointing out a mistake in a previous version of Example \ref{exe}, and for checking the new version. Work partially supported by GNFM of INDAM.


\end{document}